%% file: main.tex
\pdfoutput=1
\documentclass[conference]{IEEEtran}
\IEEEoverridecommandlockouts
\pdfoutput=1
\usepackage{amsmath,amssymb,amsfonts,amsthm,epsfig}
\usepackage[usenames,dvipsnames]{xcolor}
\usepackage{algorithm}
\usepackage{algorithmicx}
\usepackage{algpseudocode}
\usepackage{bm,xspace}
\usepackage{tcolorbox}
\usepackage{cancel}
\usepackage{fullpage}
\usepackage{liyang}
\usepackage{framed}
\usepackage{verbatim}
\usepackage{enumitem}
\usepackage{array}
\usepackage{multirow}
\usepackage{afterpage}
\usepackage{mathrsfs}
\usepackage{pifont} 
\usepackage{chngpage}
\usepackage[normalem]{ulem}
\usepackage{boxedminipage}
\usepackage{caption}
\usepackage{thmtools}
\usepackage{thm-restate}
\usepackage{float}
\usepackage{amssymb}

\def\BibTeX{{\rm B\kern-.05em{\sc i\kern-.025em b}\kern-.08em
    T\kern-.1667em\lower.7ex\hbox{E}\kern-.125emX}}
    
\usepackage[dvipsnames]{xcolor}

\def\colorful{0}

\ifnum\colorful=1

\newcommand{\blue}[1]{{{\color{blue}#1}}}
\newcommand{\red}[1]{{\color{red} {#1}}}

\fi
\ifnum\colorful=0

\newcommand{\blue}[1]{{{#1}}}
\newcommand{\red}[1]{{{#1}}}

\fi

\newcommand{\norm}[1]{\left\lVert#1\right\rVert}

\newcommand{\abs}[1]{\left\lvert#1\right\rvert}
\newcommand{\fig}[1]{\left\{#1\right\}}

\newcommand{\viol}{\mathrm{viol}}
\begin{document}

\title{Properly learning monotone functions via local correction
\thanks{Jane Lange is supported in part by NSF Award CCF-2006664, Big George Fellowship, Akamai Presidential Fellowship and Google. Ronitt Rubinfeld is supported in part by NSF awards CCF-2006664, DMS-2022448 and Fintech@CSAIL. Arsen Vasilyan is supported in part by NSF awards  CCF-2006664, CCF-1565235, CCF-1955217, Big George Fellowship and Fintech@CSAIL.}
}

\author{\IEEEauthorblockN{Jane Lange}
\IEEEauthorblockA{\textit{CSAIL} \\
\textit{MIT}\\
Cambridge, USA \\
jlange@mit.edu}
\and
\IEEEauthorblockN{Ronitt Rubinfeld}
\IEEEauthorblockA{\textit{CSAIL} \\
\textit{MIT}\\
Cambridge, USA \\
ronitt@csail.mit.edu}
\and
\IEEEauthorblockN{Arsen Vasilyan}
\IEEEauthorblockA{\textit{CSAIL} \\
\textit{MIT}\\
Cambridge, USA \\
vasilyan@mit.edu}}

\maketitle

\input{abstract}


\input{intro2.tex}

\input{preliminaries}
\input{main_result_and_consequences}
\input{correction_algorithm}

\textbf{Acknowledgments:}  We thank Mohsen Ghaffari for
discussions about his works, and the anonymous reviewers for
their helpful comments. \blue{We note that, in addition other useful discussions, Mohsen Ghaffari suggested to us the specific refinement of \Cref{remark: run-time refinement}.}

\bibliographystyle{alpha}
\bibliography{notes}

\appendix
\input{appendixb}
\input{appendixC}
\end{document}

%% file: abstract.tex
\begin{abstract}
    We give a $2^{\tilde{O}(\sqrt{n}/\eps)}$-time algorithm for properly learning monotone Boolean functions under the uniform distribution over $\{0,1\}^n$. Our algorithm is robust to adversarial label noise and has a running time nearly matching that of the state-of-the-art improper learning algorithm of Bshouty and Tamon (JACM '96) and an information-theoretic lower bound of Blais et al (RANDOM '15). Prior to this work, no proper learning algorithm with running time smaller than $2^{\Omega(n)}$ was known to exist.
    
    The core of our proper learner is a \emph{local computation algorithm} for sorting binary labels on a poset. Our algorithm is built on a body of work on distributed greedy graph algorithms; specifically we rely on a recent work of Ghaffari (FOCS'22), which gives an efficient algorithm for computing maximal matchings in a graph in the LCA model of Rubinfeld et al and Alon et al (ICS'11, SODA'12). The applications of our local sorting algorithm extend beyond learning on the Boolean cube: we also give a tolerant tester for Boolean functions over general posets that distinguishes functions that are $\eps/3$-close to monotone from those that are $\eps$-far. 
    Previous tolerant testers for the Boolean cube only distinguished between 
    $\eps/\Omega(\sqrt{n})$-close and $\eps$-far.
    
\end{abstract}

%% file: intro2.tex
\section{Introduction}

\textbf{Proper learning of monotone functions.} Consider the \emph{proper learning} problem for monotone functions:
\begin{center}
\textit{
Given i.i.d uniform labeled examples from an unknown monotone $g:\{0,1\}^n\rightarrow \{0,1\}$, output a monotone, $\eps$-accurate predictor $\hat{g}:\{0,1\}^n\rightarrow \{0,1\}$ --- that is, a circuit computing a monotone function that agrees with $g$ on at least $1-\epsilon$ fraction of the domain. 
}
\end{center}
For over 25 years there has been a large statistical-to-computational gap in our understanding of this problem. 
A $2^{\tilde{O}(\sqrt{n}/\eps)}$-time \emph{improper} learning algorithm --- that is, an algorithm that outputs a predictor $f$ that is accurate but not guaranteed to be monotone --- is given in \cite{bshouty1996fourier}.
One could use the output of this algorithm to 
obtain a monotone $\hat{g}$ by computing $f$ on every element of $\{0,1\}^n$ and solving a linear program to obtain the closest monotone function to $f$. Although this gives a $2^{\Tilde{O}\left(\sqrt{n}/\epsilon \right)}$-sample algorithm, the run-time is $2^{\Omega(n)}$. No proper learning algorithm with faster run-time (i.e. $2^{o(n)}$) was known, even given query access to $g$.

Through a new connection with local computation algorithms, we close this gap by giving a $2^{\Tilde{O}\left(\sqrt{n}/\epsilon \right)}$-time algorithm for this problem. Note that the running time essentially matches that of the aforementioned improper learning algorithm of \cite{bshouty1996fourier}. Moreover, our algorithm is essentially optimal, due to the $2^{\Tilde{\Omega}(\sqrt{n})}$ query lower bound of \cite{blais2015learning}. 
Furthermore, our algorithm is robust to adversarial noise in the labels. Specifically, in the agnostic learning model of Kearns, Schapire, and Sellie \cite{kearns1994toward}, our algorithm can handle a noise rate of $\Omega(\epsilon)$.     

\textbf{Monotonicity testing.}
The question of testing monotonicity of an unknown Boolean function over $\{0,1\}^n$ (given query access) has received a large amount of attention  
\cite{goldreich_testing_1998, DodisGLRRS99,chakrabarty_on_2013,chen_new_2014,chen_boolean_2015,khot_monotonicity_2015,belovs_polynomial_2016,chen_beyond_2017}. However, the algorithms in this line of work possess the drawback of having a large \emph{tolerance ratio}\footnote{A tester for monotonicity distinguishes a monotone function from a function that is $\epsilon$-far from monotone.
In this work, we use tolerance ratio $\alpha$ to  mean that the tester will accept a function $\epsilon/\alpha$-close to monotone.}, i.e. they will reject some functions that are extremely close to monotone. The more recent work of \cite{pallavoor2022approximating} gives a tester with tolerance ratio of $\tilde{O}(\sqrt{n})$, and this is the best tolerance ratio known\footnote{Note that Theorem 1.8 of \cite{CanonneGG0W16} gives a $2^{o(n)}$-query algorithm (no run-time bound is claimed). The run-time is still $2^{\Omega(n)}$ due to step 8 of their Algorithm 2 on page 31.} for a $2^{o(n)}$-run-time algorithm.  

As a simple corollary of our proper learning algorithm,
one can already achieve \emph{constant} tolerance ratio for monotonicity testing with  $2^{\Tilde{O}\left(\sqrt{n}/\epsilon \right)}$ run-time via a well-known connection between learning and testing \cite{goldreich1998property}. (We emphasize that, to draw this corollary, it is critical that the learning algorithm is proper and robust to noise in labels.) Even more, building more directly on our technical ideas, we present a constant-tolerance monotonicity tester with \emph{exponentially better dependence on $\epsilon$ of $2^{\Tilde{O}\left(\sqrt{n \log \frac{1}{\epsilon}}  \right)}$}. This approach also yields a tolerant monotonicity tester for functions over general posets, which we describe in detail in \Cref{cor:poset tester}. By a standard reduction \cite{Parnas2004TolerantPT}, this also gives an approximator for distance to monotonicity (within a constant multiplicative error plus an $\eps N$ additive error) in functions over general posets as well.



\subsection{Monotonicity correction via sorting on partially ordered sets.}

\textbf{The poset-sorting problem.} One of our core ideas is to create a new bridge between our learning task and recent exciting developments on parallelizing greedy algorithms \cite{Ghaffari16, GhaffariU19, Ghaffari22}. These developments accomplish speedups for classic graph problems in the setting of local and distributed computing. They build on a large body of work in local computation algorithms and distributed graph algorithms: some examples include
\cite{RTVX11, AlonRVX12, levi_local_2015, GoosHLMS15, 
reingold_new_2014, even_best_2014, ChangFGUZ19, EvenLMR21, ParterRVY19, ArvivL21, BrandtGR21, GrebikR21}.
In order to accomplish this, we introduce the problem of poset sorting, that is, sorting binary values on a partially ordered set (poset), which we believe is of independent interest: 
\begin{center}
\textit{
Let $P$ be a poset of $N$ elements with longest chain length at most $h$ and such that every element has at most $\Delta$ predecessors or successors. Given a binary labeling $f$ of elements of $P$, output a new binary labeling $f_\text{mon}$ that (i) is monotone with respect to the partial order in $P$ (ii) can be obtained from $f$ by a sequence of swaps of monotonicity-violating label pairs. 
}
\end{center}
Clearly, there is a greedy algorithm for this task that keeps swapping monotonicity-violating pairs of labels until there are none left. 
The challenge is to do this in a distributed fashion.
Among the many distributed and local computation models, the one that turns out relevant to us is the {\sl local computation algorithm} (LCA) model, defined formally in \Cref{sec:prelim}.
In brief, to be an LCA, an algorithm should be local in the sense that it should not need to read the entire assignment of labels to vertices in order to determine which label will end up at some particular vertex $x$. It should suffice to read only the labels of vertices that lie in a restricted neighborhood of $x$ --- ideally not much more than the set of vertices one would need to read in order to determine whether $x$'s own label violates monotonicity. We explain the algorithm itself in Subsection \ref{subsection: the algorithm summary}. 

\textbf{Proper learning via local correction.} 
Why is an algorithm for sorting on a poset $P$ relevant to proper learning? Suppose we obtained an improper predictor $f$ via the algorithm of \cite{bshouty1996fourier}, in the form of a small circuit computing $f$. Then, let monotone $f_\text{mon}$ be obtained from $f$ by flipping a sequence of monotonicity-violating labels.
It is not hard to argue that since $f$ is $O(\epsilon)$-close to the monotone $g$ we are trying to learn, then so is any $f_\text{mon}$ obtained in this manner (see \Cref{remark: sorting implies local correction} for more details). What the LCA allows us to do is to transform the circuit computing $f$ into a \emph{small} circuit computing such $f_\text{mon}$. 
The reason is that to evaluate $f_\text{mon}$ at a given element $x$ the LCA evaluates $f$ on only a small number of points and has an appropriately fast run-time. This allows us to augment the circuit for $f$ with a circuit that executes this sorting algorithm, and as a result obtain a small circuit for $f_\text{mon}$.

In other words, an LCA for sorting on a poset is a {\sl local corrector} for monotonicity: an algorithm that takes some input $x$, makes queries to a black-box function $f$, and outputs $f_{mon}(x)$ where $f_{mon}$ is a monotone function that is close to $f$ in Hamming distance, if such a function exists. 
Examples of local correctors for various function properties
can be found in 
\cite{AilonCCL08,  ACCL07, SS10, KalePS08,  CampagnaGR13, JhaR13, AJMR14, Comandur16}. A local corrector, combined with any improper hypothesis, yields a proper hypothesis. The efficiency of the LCA determines how quickly the proper hypothesis can be evaluated.



\subsection{Our LCA for sorting on a poset.}\label{subsection: the algorithm summary}

When the longest chain length $h$ is 1, then the poset-sorting problem is equivalent to the classical problem of finding a maximal matching on a bipartite graph with $N$ vertices and maximum degree at most $\Delta$. We note that a recent LCA by Ghaffari \blue{\cite{Ghaffari22}} handles this problem using a run-time of only \blue{$
\poly(\Delta, \log(N/\delta))$}. 

For larger values of $h$, a naive approach would be to execute a sequence of phases, in each of which a maximal matching between monotonicity-violating labels is produced and the labels that are matched with each other are swapped. One can show that $O(h)$ such phases suffice and sometimes necessary if the matchings are arbitrary. If the algorithm by Ghaffari \cite{Ghaffari22} is used to implement each of the phases, this yields a run-time of $\left(\blue{
\Delta \log(N/\delta)} \right)^{\Theta(h)}$.
For properly learning monotone functions over the Boolean cube, the parameters\footnote{These are the parameters for small constant $\epsilon$. A slight subtlety in our argument is that we need to work over truncated hypercube, i.e. handle separately $O(\epsilon)$ fraction of the elements with too high or too low Hamming weight. This is a standard technique, and in our case it makes the parameter $\Delta$ in our poset small enough for us to get a good run-time guarantee using our poset sorting algorithm.
} we are primarily interested in are $h =\Theta(\sqrt{n})$, $\Delta = 2^{\tilde{\Theta}(\sqrt{n})}$ and $N=\Theta(2^n)$.
The naive approach then gives us a run-time of $\left(\blue{
\Delta \log(N/\delta)} \right)^{\Theta(h)}=2^{\Omega(n)}$, which is too slow for our purposes.

We beat the naive approach by enforcing that the maximal matchings at each step only include pairs of vertices that are sufficiently far away in the graph.
We show that after each matching step, the greatest distance between pairs that violate monotonicity reduces by a factor of 2. This allows us to reduce the exponent from $h$ to $\log h$, which is sufficiently fast to yield an essentially optimal proper learning algorithm for monotone functions.




\subsection{Other related work}
The problem of locally correcting monotonicity has been studied in \cite{AilonCCL08,ACCL07, SS10, BGJ+10, AJMR14} 
in various parameter settings. 
The work of \cite{AilonCCL08} introduces the problem of online property reconstruction and gives an algorithm for correcting monotonicity for real-valued functions over the discrete number line. The work of \cite{Parnas2004TolerantPT} gives a tolerant tester in the same setting. The work of \cite{SS10} introduces the framework of \emph{local} property reconstruction, which is the same framework our approach uses (i.e. local correction by memoryless LCA). They give a local corrector for functions over the hypergrid $[d]^n$, with large dependence on the dimension $n$ but small dependence on $d$. Lower bounds for monotonicity correction in other error regimes are given in \cite{BGJ+10, AJMR14}. 
The problem of approximating the distance to monotonicity, which is strongly related to tolerant testing, has been studied in \cite{pallavoor2022approximating, ACCL07, Parnas2004TolerantPT}. 

A proper learning algorithm for a function class that generalizes monotone functions is given in \cite{CanonneGG0W16}.
Proper learning of restricted classes of monotone functions has been studied in \cite{Servedio2009OnTL,Ang88,YBC13,Servedio2009OnTL, Blanc2022ProperlyLD}. 
The question of weak learning of monotone functions has also received attention \cite{KearnsV89,bshouty1996fourier,BlumBL98,AmanoM02,ODonnellW09}. The latter line of work investigates proper learning algorithms that have very fast run-time at the cost of having accuracy of only $\frac{1}{2}+\frac{1}{\poly(n)}$.

In addition to the Boolean cube, testing monotonicity has also been studied on hypergrids, see for example
\cite{ChakrabartyS13a},
\cite{BlaisRY14}
\cite{ChakrabartyS13},
\cite{BlackCS18},
\cite{BlackCS20}. Also see
\cite{ChakrabartyS19}
for monotonicity testing of functions with bounded influence.


\subsection{Organization of this paper}
In \Cref{sec:LCA} we define the LCA model and state the maximal matching result of \cite{Ghaffari22}.
In \Cref{sec:main} we state and prove the main proper learning and testing results as consequences of our local poset sorting algorithm. 
In \Cref{sec:sorting} we present pseudocode for the local sorting algorithm and analyze its correctness and complexity.

%% file: preliminaries.tex
\section{Preliminaries}
\label{sec:prelim}

\subsection{Notation (posets and distances)}

Let $x$ and $y$ be elements of a poset $P$.\footnote{In this work all posets are assumed to be finite.} We use $\preceq$ to denote the ordering relation on $P$. We say $x\prec y$ (``$x$ is a predecessor of $y$'') if $x\preceq y$ and $x\neq y$. Also, we say $x\succeq y$ if $y\preceq x$, and $x\succ y$ (``$x$ is a successor of $y$'') if $y\prec x$. We say $x$ and $y$ are \emph{incomparable} if neither $x\prec y$ nor $x\succ y$ holds.
We say that $x$ is an \emph{immediate predecessor} of $y$ if $x\prec y$ and there is no $z$ in $P$ for which $x\prec z \prec y$. The notion of an \emph{immediate successor} is defined analogously.

The \emph{graph} of a poset $P$ (a.k.a. \emph{Hasse diagram} of $P$) is a directed graph, in which elements of $P$ are the vertices, and there is an edge from $x$ to $y$ whenever $x$ is an immediate predecessor of $y$. 
Clearly, the graph of any poset is a DAG. Additionally, it is immediate that the poset itself is unambiguously determined by its graph, and we will refer to the poset and its graph interchangeably. 
\begin{definition}[Graph distance and height]
We will write $\dist(x,y)$ to denote the length of the \textbf{longest} directed path\footnote{If $x$ and $y$ are incomparable, then $\dist(x,y)$ is undefined.} between $x$ and $y$ in the Hasse diagram of $P$. The \textbf{height} of a poset $P$ is the length of the longest directed path between any two elements in $P$. 
\end{definition}
\begin{definition}[Function distance]
For a pair of functions $f_1, f_2: P\rightarrow \{0, 1\}$, the \textbf{distance} $\norm{f_1-f_2}$ is the fraction of elements $x$ in $P$ on which $f_1(x)\neq f_2(x)$.
\end{definition}
\begin{definition}[Distance to monotonicity]
For a function $f:P\rightarrow \{0, 1\}$, the \textbf{distance to monotonicity} is defined as $\min_{\text{monotone Boolean }f'}\left(\norm{f-f'}\right)$.
\end{definition}
\subsection{Agnostic learning setting}

Now, we formally describe the setting of agnostic learning under the uniform distribution. The learning algorithm is given i.i.d. example-label pairs from some distribution over $\{0,1\}^n \times \{0,1\}$, where the marginal distribution over examples is uniform. The \emph{generalization error} (which we also refer to simply as \emph{error}) of a predictor is the probability it misclassifies a fresh example-label pair. The goal of agnostic learning is to produce a predictor $\hat{g}$ with good generalization error. Guarantees are produced in terms of the lowest generalization error among all hypotheses in the function class (in our case monotone functions), which we denote as $\opt$.


\subsection{The LCA model}
\label{sec:LCA}

The Local Computation Algorithm (LCA) model captures the 
ability to provide query access to parts of an output in sublinear time.
In this work, we use the LCA model for the 
problems of maximal matching and 
poset sorting.   
An LCA is given access to a random bit-string
and to the input:  
for example, in the case of maximal matching,
this input is the adjacency list of the graph.
Upon receiving an edge query $e$ in $G$,
the LCA should respond ``yes'' or ``no.''
Responses to different
edges must be consistent with a single legal
maximal matching.  
Similarly, for poset sorting on poset $P$, 
a query to the LCA is any element $x \in P$, 
and the LCA, which is provided with access to
function $f:P \rightarrow \{0,1\}$, must respond with  $f_{\text{mon}}(x)$ so that
the function $f_{\text{mon}}$ is monotone.
Responses to different inputs must be consistent with 
a single legal sorting of $f$.
In both cases,  the LCA's answers
may depend on a random
bit-string\footnote{In all cases we consider, the random bit-string is short enough for the LCA to read as a whole (in contrast, in some work exponentially long bit-strings are considered). 
} that persists between queries, but otherwise, responses must be\footnote{Sometimes LCAs are considered that are not memoryless, but in this work whenever we refer to an LCA, we imply it is memoryless.}
{\em memoryless} -- i.e., they cannot depend on previous queries to or responses
from the LCA.

%
In terms of performance, we want the following quantities to be as small as possible (i) the \emph{query complexity}, i.e. the number of probes to the input the LCA makes to respond to a single query, (ii) the \emph{run-time} the LCA needs to respond to a query, (iii) the length of the random bit-string used by the LCA, (iv) the probability over the random bit-string that the LCA fails to satisfy the problem specifications. 

We will use a recent powerful result\footnote{We need to comment on some superficial differences between Theorem \ref{thm: Ghaffari Uitto LCA theorem} and the main theorem of \blue{\cite{Ghaffari22}}, which does not mention run-time and bit-string length explicitly and also has \blue{$\delta=\poly\left(\frac{1}{N}\right)$}. The following observations are not novel in any way, and some of them are alluded to in \blue{\cite{GhaffariU19, Ghaffari22}}, but we explain them here for completeness.
The bound on run-time follows by direct inspection of their algorithm. 
The length of the bit-string can be reduced to \blue{$
\poly(\Delta, \log(N/\delta))$} via the standard method \cite{AlonRVX12,levi_local_2015} of replacing i.i.d. random bits with $k$-wise independent random bits for $k$ large enough ($k=\blue{
\poly(\Delta, \log(N/\delta))}$ suffices, as number of those i.i.d. random bits accessed per query is \blue{$
\poly(\Delta, \log(N/\delta))$}). Although, the failure probability bound $\delta$ in the original theorem of \cite{Ghaffari22} is set to be \blue{$\poly\left(\frac{1}{N}\right)$}, it can be boosted to arbitrary $\delta$ by adding extra disconnected vertices to our graph until \blue{this value reaches} $\delta$ for the new number of vertices $N'$. Query access to this new graph can be simulated via query access to the original graph with inconsequential overhead. Overall, this costs one an extra $\polylog(1/\delta)$ factor in
query complexity and run-time.
}
of \blue{Ghaffari \cite{Ghaffari22}}, which gives an efficient algorithm for answering membership queries to a maximal independent set. (We emphasize that \emph{maximal} independent set is defined to be an independent set that cannot be made into a larger independent set by adding an extra vertex, and \emph{maximal matching} is defined analogously.) \blue{We note that
in an earlier version of this paper (which was written before \cite{Ghaffari22} was available) we used the theorem of Ghaffari and Uitto  \cite{GhaffariU19} for this purpose.   }
\begin{theorem}[\blue{\cite{Ghaffari22}}]\label{thm: Ghaffari Uitto LCA theorem}
There is an LCA that takes all-neighbor\footnote{I.e. when queried a vertex $v$, the oracle returns all the neighbors of $v$.} access to a graph $G$ with $N$ vertices and largest degree at most $\Delta$, and gives membership access to a \emph{maximal independent set} on the graph. The query complexity and the run-time of the LCA are \blue{$
\poly(\Delta, \log(N/\delta))$}, the length of the random bit-string is also \blue{$
\poly(\Delta, \log(N/\delta))$} and the failure probability is at most $\delta$.
\end{theorem}

\begin{corollary}
\label{cor:matchings}
There is an LCA that takes all-neighbor access to a graph $G$, with $N$ vertices and largest degree at most $\Delta$, and gives membership access to a \emph{maximal matching} on the graph. The query complexity and the run-time of the LCA is \blue{$
\poly(\Delta, \log(N/\delta))$}, the length of the random bit-string is also \blue{$
\poly(\Delta, \log(N/\delta))$} and the failure probability is at most $\delta$.
\end{corollary}
\begin{proof}
This reduction is standard; see \Cref{sec: proof of cor:matchings} for details.
\end{proof}

%
%
\subsection{Boolean hypercube.}
\begin{definition}
The \emph{$n$-dimensional Boolean hypercube} is the set $\{0,1\}^n$.
For $x,y\in \{0,1\}^n$, we say $x\preceq y$ if for all $i\in\{1,\cdots,n\}$ one has $x_i\leq y_i$. It is immediate that $\{0,1\}^n$ is a poset with $2^n$ elements. 

We also define the truncated hypercube 
$$
H^n_\epsilon:=\fig{x \in \{0,1\}^n:~\abs{\sum_ix_i-\frac{n}{2}}\leq \sqrt{\frac{n}{2}\log \frac{2}{\epsilon}}},
$$
Via Hoeffding's bound, we have that the fraction of elements in $\{0,1\}^n$ that are not also in $H_n^\epsilon$ is at most $2\exp\left(-\frac{2t^2}{n}\right)=\epsilon$.
\end{definition}

\textbf{Known results about learning monotone functions over Boolean hypercube.}
\begin{theorem}[Learnability of monotone functions \cite{bshouty1996fourier}]\label{thm: bshouty-tamon realizable}
There is an algorithm that, for any monotone function $g:\{0,1\}^n\rightarrow \{0,1\}$, given i.i.d. example-label pairs $(x_i,g(x_i))$, with $x_i$ uniform in $\{0,1\}^n$, returns a circuit computing a predictor $\hat{g}$, such that $\norm{g-\hat{g}}\leq \epsilon$.
The algorithm uses $\blue{n^{O\left(\frac{\sqrt{n}}{\epsilon}\right)}}\log\left(\frac{1}{\delta}\right)$ samples and run-time, where $\delta$ is the failure probability bound.
\end{theorem}

The standard theorem below follows via low-degree concentration result of \cite{bshouty1996fourier}, Remark 4 on page 6 of \cite{kalai_agnostically_2005}, the refinement\footnote{We note that in an earlier version of this work we were not aware of the refinement of \cite{feldman_tight_2020}, so the $\epsilon$-dependence in the \Cref{thm: bshouty-tamon agnostic}, as well as in \Cref{thm: proper learning semi-agnostic}, was slightly worse.} of \cite{feldman_tight_2020} and standard failure probability reduction via repetition. Please refer to \Cref{sec: best improper  agnostic learning algorithm from the literature.} for more detail.
\begin{theorem}[Agnostic learnability of monotone functions]\label{thm: bshouty-tamon agnostic}
In the agnostic setting with examples distributed uniformly on $\{0,1\}^n$, there is an algorithm that returns a circuit with generalization error at most $\textnormal{opt}+\epsilon$, where $\textnormal{opt}$ is the error of the best monotone function. The algorithm uses $\blue{n^{O\left(\frac{\sqrt{n}}{\epsilon}\log \frac{1}{\epsilon}\right)}}\log\left(\frac{1}{\delta}\right)$ samples and run-time, where $\delta$ is the failure probability bound.

\end{theorem}

%% file: main_result_and_consequences.tex
\section{Main result and consequences}
\label{sec:main}

We first present the formal statement for our LCA for the poset sorting problem, from which every other result in this section is derived. The algorithm and analysis are presented in \Cref{sec:sorting}. 
\begin{restatable}{theorem}{main}
\label{thm: the local corrector for general posets}
Let $P$ be a poset of $N$ elements and height $h$, such that each element in $P$ has at most $\Delta$ predecessors or successors.
Suppose we are given query access to the graph of $P$, i.e. for any $x \in P$ we can obtain the immediate predecessors or successors of $x$.
Also, suppose we are given query access to some Boolean function $f$ over $P$.
Then, there is an LCA that solves the poset-sorting problem for $f$ over $P$: in other words, it provides query access to a monotone function $f_{mon}: P \to \zo$ that can be obtained from $f$ by a sequence of swaps of monotonicity-violating label pairs.
The LCA has query complexity and run-time of $\left(\Delta \log\left(\frac{N}{\delta}\right)\right)^{O(\log h)}$, and it uses a random bit-string of length $\poly(\Delta \log \frac{N}{\delta})$.
The failure probability of the LCA is at most $\delta$.

\end{restatable}

\begin{restatable}{remark}{refinement remark}
\label{remark: run-time refinement}
In the setting of Theorem \ref{thm: the local corrector for general posets}, 
suppose each element of $P$ has at most $n$ immediate predecessors or successors,
and furthermore, that $P$ is graded 
(all paths between $u$ and $v$ for any $u,v$ have the same length). 
Then query complexity and the run-time of our LCA is also bounded by $n^{O(h)}  (\log(N/\delta))^{\log h}$.
\blue{The number of 
random bits used is at most  $n^{O(h)}\polylog(N/\delta)$.}
In particular, when $P$ is the truncated hypercube $H_\eps^n$ and $\delta = 2^{-10n}$,
the query and time complexity are $n^{O\left(\sqrt{n \log \frac{1}{\eps}}\right)}$.
\end{restatable}

The proof of \Cref{remark: run-time refinement} is given in \Cref{sec: refined local implementation}.

\begin{proposition}[Local correction]
\label{remark: sorting implies local correction}
In the setting of the previous problem, the distance between $f$ and $f_\text{mon}$ is at most twice the distance of $f$ to monotonicity. Furthermore, the following extra property holds: for any monotone $q:P\rightarrow  \{0,1\}$ we have $\norm{f_\text{mon}-q} \leq \norm{f-q}$. 
\end{proposition}
\begin{proof}
We first show the extra property. Recall that $f_\text{mon}$ can be obtained from $f$ via a sequence of swaps of monotonicity-violating labels. Since $q$ is monotone, as a result of every single of this swaps the distance to $q$ will either decrease or stay the same. Overall across all the swaps, this means that $\norm{f_\text{mon}-q} \leq \norm{f-q}$.

Taking $q$ to be the closest monotone function to $f$ and using the triangle inequality, we see that the distance between $f$ and $f_\text{mon}$ is at most twice the distance of $f$ to monotonicity.
\end{proof}
The two corollaries about tolerant testing of monotone functions follow from our theorem above. We note that the success probability of $2/3$ can be improved via repetition to $1-\delta$ at the cost of $\log \left(\frac{1}{\delta}\right)$ multiplicative factor in run-time and query complexity. We also note that the inverse tolerance ratio, given as $0.49$, can be improved to any absolute constant less than $0.5$.

\begin{corollary}[Tolerant testing for the Boolean cube]
\label{cor:cube tester}
Suppose we are given query access to an unknown Boolean function $f$. Then, there is an algorithm that uses $\blue{n^{O\left(\sqrt{n\log\frac{1}{\epsilon}}\right)}}$ queries and run-time, and distinguishes whether the function $f$ is $0.49\epsilon$-close or $\epsilon$-far from monotone.
The failure probability is at most $2/3$.
\end{corollary}
\begin{proof}

We use the truncated hypercube $H^n_{0.005\epsilon}$ as our poset when using
Theorem \ref{thm: the local corrector for general posets} (also using the refined run-time of \Cref{remark: run-time refinement}). This allows us to gain query access to a monotone $f_\text{mon}$ on $H^n_{0.005\epsilon}$. Extend $f_\text{mon}$ to all of $\{0,1\}^n$ by setting it to $1$ above the upper truncation threshold and to $0$ below the lower threshold. Clearly, $f_\text{mon}$ is now also monotone on all of $\{0,1\}^n$.

We sample i.i.d. uniformly random elements of $\{0,1\}^n$, evaluate both $f$ and $f_\text{mon}$ on each these elements and obtain an estimate of $\norm{f-f_\text{mon}}$ up to error $0.005\epsilon$. 
If $f$ is $\epsilon$-far from monotone, then the distance $\norm{f-f_\text{mon}}$ is also at least $\epsilon$, so the estimate will be at least $0.995\epsilon$. If $f$ is $0.49 \epsilon$-close to monotone, then there is some monotone function over $H^n_{0.005\epsilon}$ with which $f$ disagrees on at most $0.49\epsilon \cdot 2^n$ elements of $H^n_{0.005\epsilon}$. The guarantee of Theorem \ref{thm: the local corrector for general posets} (via \Cref{remark: sorting implies local correction}) tells us that then $f$ and $f_\text{mon}$ disagree on at most $0.98\epsilon \cdot 2^n$ elements of $H^n_{0.005\epsilon}$. Since there are only at most $0.005\epsilon \cdot 2^n$ elements in  $\{0,1\}^n$ that are not in $H^n_{0.005\epsilon}$, we see that $\norm{f-f_\text{mon}}$ is at most $0.985\epsilon$. Therefore, the estimate will be at most $0.99\epsilon$. Overall, checking if the estimate is greater than $0.992\epsilon$ allows us to distinguish whether $f$ is $\epsilon$-far from monotone or $0.49\epsilon$-close to monotone.

 For the estimation to succeed, we need to evaluate $f$ and $f_\text{mon}$ on $O\left(\frac{1}{\epsilon^2}\right)$ i.i.d. random elements of $\{0,1\}^n$. For the LCA of $f_\text{mon}$ we can set the overall success probability parameter to be $0.1$. A Chernoff bound and union bound argument then shows that overall success probability is at least $2/3$.
 \blue{For $H^n_{0.005\epsilon}$ our parameters are $h=O\left(\sqrt{n\log\frac{1}{\epsilon}}\right)$, $N=O(2^n)$ and each element of the poset has at most $n$ immediate predecessors or successors. Overall, the run-time given by \Cref{remark: run-time refinement} is 
 $
  n^{O\left(\sqrt{n\log\frac{1}{\epsilon}}\right)}.
 $}
 
\end{proof}

\begin{corollary}[Tolerant testing for general posets]
\label{cor:poset tester}
Suppose we are
in the setting of Theorem \ref{thm: the local corrector for general posets}, and we also have access to an oracle giving us i.i.d. uniform elements in $P$. Then, there is an algorithm that uses\\ $\Delta^{O(\log h \log\log \Delta)}\left(\log\left(N\right)\right)^{O(\log h)}\frac{1}{\epsilon^2}$ queries and run-time, and distinguishes whether the function $f$ is $0.49\epsilon$-close or $\epsilon$-far from monotone.
The failure probability is at most $2/3$.
\end{corollary}
\begin{proof}
The proof is similar to the proof of \Cref{cor:cube tester}, and is given in \Cref{sec: proof of cor:poset tester}. 
\end{proof}

\begin{theorem}[Proper learnability of monotone functions]\label{thm: proper learning realizable}
There is an algorithm that, for any monotone function $g:\{0,1\}^n\rightarrow \{0,1\}$, given i.i.d. example-label pairs $(x_i,g(x_i))$, with $x_i$ uniform in $\{0,1\}^n$, returns a circuit computing a \textbf{monotone function} $\hat{g}$, such that $\norm{g-\hat{g}}\leq \epsilon$.
The algorithm uses \blue{$n^{O\left(\frac{\sqrt{n}}{\epsilon}\right)}$} samples and run-time and fails with probability at most $1/2^n$.
\end{theorem}

\begin{proof}
The proper learner does the following:
\begin{enumerate}
    \item Use the improper learner in Theorem \ref{thm: bshouty-tamon realizable} with error parameter $\frac{\epsilon}{10}$ and failure probability bound $1/2^{n+1}$. This gives a circuit computing a function $f$ over $\{0,1\}^n$.
    \item Obtain a circuit computing a function $f_\text{mon}:H^n_{\epsilon/10}\rightarrow\{0,1\}$ as follows. The circuit is computed via running the LCA from Theorem \ref{thm: the local corrector for general posets} with accuracy parameter equal to $\frac{\epsilon}{10}$ and failure probability bound equal to $1/2^{n+1}$, and with the oracle calls to a function replaced with an evaluation of the circuit $f$, restricted to $H^n_{\epsilon/10}$. The random bit-string used by the LCA is hard-coded into the circuit for $h$, so that the resulting circuit is deterministic.
    \item Augment the circuit computing $f_\text{mon}$ in order to extend this function into the whole of $\{0,1\}^n$ as follows. If $|x|>\frac{n}{2}+\sqrt{\frac{n}{2}\log \frac{20}{\epsilon}}$ then $f_\text{mon}(x)=1$, and if $|x|<\frac{n}{2}-\sqrt{\frac{n}{2}\log \frac{20}{\epsilon}}$ then $f_\text{mon}(x)=0$.
    \item Output the circuit computing $f_\text{mon}$.
\end{enumerate}
With probability at least $1/2^{n}$ both of the algorithms we invoke succeed, which we will assume henceforth. 

\blue{
Let us discuss the run-time. Step 1 runs in time $n^{O\left(\frac{\sqrt{n}}{\epsilon}\right)}$ and the circuit for $g$ can therefore only have size at most $n^{O\left(\frac{\sqrt{n}}{\epsilon}\right)}$. For step 2, first observe that we have $\left\lvert H^n_{\epsilon/10} \right\rvert\leq 2^n$, $H^n_{\epsilon/10}$ has height $2\sqrt{\frac{n}{2}\log \frac{20}{\epsilon}}$ and each element in $H^n_{\epsilon/10}$ can have only at most $n$ immediate predecessors and successors. Therefore, the LCA from Theorem \ref{thm: the local corrector for general posets} (refined via \Cref{remark: run-time refinement}) in this setting has run-time, query complexity and bit-string length of $n^{O(\sqrt{n}\log(1/\epsilon))}$. Since the circuit for $f$ itself has size $n^{O\left(\frac{\sqrt{n}}{\epsilon}\right)}$, the overall run-time of the learning algorithm and the size of circuit computing $f_\text{mon}$ is also $2^{O\left(\frac{\sqrt{n}}{\epsilon}\right)}$.}

Finally, we argue correctness. Correctness of the LCA in Theorem \ref{thm: the local corrector for general posets} implies that $f_\text{mon}$ is monotone over $H^n_{\epsilon/10}$ and we see that the extension of this function to $\{0,1\}^n$ in step $3$ keeps it monotone. 

Now, let $g$ be the function we are trying to learn. Since we are in the realizable setting, $g$ is monotone. 
Theorem \ref{thm: bshouty-tamon realizable} tells us that $f$ and $g$ disagree on at most $\frac{\epsilon}{10}2^n$ elements. This, together with Theorem \ref{thm: the local corrector for general posets}, Remark \ref{remark: sorting implies local correction} and the fact that $g$ is monotone, tells us that $f_\text{mon}$ disagrees with $f$ on at most $\frac{\epsilon}{5}2^n$ elements of $H^n_{\epsilon/10}$. The number of $x\in \{0,1\}^n$ not in $H^n_{\epsilon/10}$ is at most $\frac{\epsilon}{10}2^n$, so overall $f_\text{mon}$ disagrees with $f$ on at most $\frac{3\epsilon}{10}2^n$ elements of $\{0,1\}^n$, in other words $\norm{f-f_\text{mon}} \leq \frac{3\epsilon}{10}$. Via triangle inequality, we have $\norm{g-f_\text{mon}} \leq \norm{g-f} +\norm{f-f_\text{mon}} \leq \frac{2\epsilon}{5}\leq \epsilon$, finishing the proof.
\end{proof}
Let us remark on the performance of our algorithm in the agnostic setting.
Observation 3 on page 5 of \cite{kalai_agnostically_2005} implies that the algorithm of \cite{bshouty1996fourier} (i.e. Theorem \ref{thm: bshouty-tamon realizable}), when run in the agnostic setting, will give a predictor with error at most $8\cdot \text{opt}+\epsilon$, where $\text{opt}$ is the error of best monotone predictor. Repeating the argument in the proof of Theorem \ref{thm: proper learning realizable} then tells us that in this setting our proper learning algorithm will also have prediction error $C\cdot \text{opt}+\epsilon$ for some absolute constant $C$. In particular, this means\footnote{Let us elaborate. If we take the error parameter in our algorithm to be $\epsilon/2$, then we see that our algorithm will have prediction error $C \cdot \text{opt}+\epsilon/2$. Then, if noise rate $\text{opt}$ is $\frac{\epsilon}{2C}$ or less, our predictor will be $\epsilon$-competitive with the best monotone function, as required by the agnostic learning model of Kearns, Schapire, and Sellie \cite{kearns1994toward}. } in the agnostic learning model of Kearns, Schapire, and Sellie \cite{kearns1994toward}, our algorithm can handle a noise rate of $\Omega(\epsilon)$. 

We now present how to obtain a better error guarantee in the agnostic setting at a cost of slightly worse dependence of run-time on $\epsilon$:      
\begin{theorem}[Proper learning in agnostic setting]\label{thm: proper learning semi-agnostic}

In the agnostic setting with examples distributed uniformly over $\{0,1\}^n$, 
there is a learning algorithm that outputs a circuit computing a \textbf{monotone function} $\hat{g}$, such that if the best monotone predictor has error $\opt$, then the error of the predictor $\hat{g}$ is at most $3 \cdot \opt +\epsilon$. 
The algorithm uses \blue{$n^{O\left(\frac{\sqrt{n}}{\epsilon} \log\frac{1}{\epsilon}\right)}$} samples and run-time. The failure probability of the algorithm is at most $1/2^n$.
\end{theorem}
\begin{proof}
The proof, presented in  \Cref{sec: proof of thm: proper learning semi-agnostic},
follows a pattern similar to the proof of Theorem \ref{thm: proper learning realizable}. 
\end{proof}

%% file: correction_algorithm.tex
\section{The LCA for poset sorting}
\label{sec:sorting}
In this section, we prove \Cref{thm: the local corrector for general posets}. First, we give a ``global'' algorithm for the poset sorting problem, which reads all the values of $f$ and writes all the values of $f_{mon}$. The global algorithm is inefficient, but lends itself easily to a proof of correctness. We prove correctness for the global algorithm, then give our local implementation and show that it simulates the global algorithm.

\subsection{A global view} 
\label{subsection: the global view}
We first present~\Cref{alg:global}, which sorts the labels of $f$ in stages by swapping the labels of pairs of vertices that violate monotonicity. We will show that each stage reduces the maximum distance between violated vertices by a factor of 2, which produces a monotone function after $\log h$ stages, where $h$ is the height of the input graph.

Before we present the algorithm, we define the following objects that it constructs during its execution.

\begin{definition}[Violation set]
We define the set of violated pairs $\mathrm{viol}_P(f)$ as follows:
\begin{align*}
    &\viol_P(f) := \\ &\{(v, w) \in P \times P : ~ v \succ w,~ f(v)=0 \text{ and } f(w)=1\}.
\end{align*}
\end{definition}

\begin{definition}[$k$-violation graph]
\label{def:viol-graph}
For a poset $P$ of height $h$, a function $f~:~P\to \zo$, and some $k \in [h]$, we define the $k$-violation graph $B_k$ as follows:

\begin{itemize}
	\item $V(B_k) = P$, and 
	\item For $(v, w) \in \mathrm{viol}_P(f)$, $(v, w) \in E(B_k)$ iff $\dist(v,w) \ge k$. 
\end{itemize}
Note that $B_k$ is bipartite and undirected.
\end{definition}

\begin{algorithm}[H]
\begin{algorithmic}
\State \textbf{Given:} Poset $P$ of height $h$, function $f_0:~P \to \zo$
\State \textbf{Output:} monotone function over $P$
\State Let $i \gets 0$
\For{$0 \le i \le \lceil \log h \rceil + 1$}
\State Let $k \gets \lceil h/2^{i+1} \rceil$
\State Construct the $k$-violation graph $B_k$ from $P$ and $f_i$.
\State Compute a maximal matching in $B_k$ and let $\lambda$ map matched vertices to each other, and unmatched vertices to themselves. 
\State For all $x \in P$, let $f_{i+1}(x) = f_i(\lambda(x))$
\State $i \gets i + 1$
\EndFor
\State \Return $f_i$
\end{algorithmic}
\caption{LCA for sorting labels in a poset: the global view}
\label{alg:global}
\end{algorithm}

\subsection{Correctness of \Cref{alg:global}}
Recall that we are required to show that our algorithm outputs a monotone function that can be obtained from $f$ by a sequence of monotonicity-violating label swaps.
Since it is evident from the pseudocode that this algorithm only performs such swaps, it remains to show only that the output is monotone.

\input{path_shortening_general}

The following invariant, which will be useful for proving that the output is monotone, is a consequence of \Cref{lemma: path shortening lemma}.

\begin{corollary}[Distance shortening invariant]
\label{claim:distance-invariant}
The following holds for all $f_i$, $0 \le i \le \lceil \log h \rceil$:
$$
\max_{(v, w) \in \viol_P(f)} \dist(v, w) \leq \bigg \lceil \frac{h}{2^i} \bigg \rceil
$$
\end{corollary}

\begin{proof}
We proceed by induction on $i$. For $i=0$, the distance must be at most $h$, because $h$ is the height of $P$. Assume as an inductive hypothesis that the claim holds for some $i \le \lceil \log h \rceil$.

$B_k$ is a graph with the properties described in~\Cref{def:viol-graph}: it has an edge joining each pair of vertices in $P$ that violates monotonicity and has distance at least $k$, for $k = \lceil h / 2^{i+1} \rceil$. By the inductive hypothesis, all such distances are between $\lceil h/2^{i+1}\rceil $ and $\lceil h/2^i \rceil$ inclusive. Then by~\Cref{lemma: path shortening lemma} we guarantee that 
\begin{align*}
&\max_{(v, w) \in \viol_P(f)} \dist(v, w) \\
&\leq \max \bigg (\bigg\lceil \frac{h}{2^{i+1}} \bigg\rceil - 1,  ~ \bigg\lceil \frac{h}{2^i} \bigg\rceil - \bigg\lceil \frac{h}{2^{i+1}}\bigg\rceil \bigg ) \\
&\le \bigg\lceil \frac{h}{2^{i+1}}\bigg\rceil .
\end{align*}
This completes the induction.
\end{proof}
With~\Cref{claim:distance-invariant} in hand, we continue with the proof of correctness.

\begin{lemma}[Correctness]
		\label{lem:correctness}
For any Boolean function $f_0$, poset $P$, and any choice of maximal matchings over $B_k$ in \Cref{alg:global}, the output $f_{\lceil \log h \rceil + 1}$ is monotone over $P$.
\end{lemma}

\begin{proof}
By~\Cref{claim:distance-invariant}, we have 
$$
\max_{(v, w) \in \viol_P(f_{\lceil \log h \rceil})} \dist(v, w) \leq \bigg \lceil \frac{h}{2^{\lceil \log h \rceil }} \bigg \rceil = 1
$$
Then $f_{\lceil \log h \rceil }$ has the property that all pairs that violate monotonicity are immediate neighbors in $P$. By one more application of \Cref{lemma: path shortening lemma}, we have 
\begin{align*}
&\max_{(v, w) \in \viol_P(f_{\lceil \log h \rceil + 1})} \dist(v, w) \\
&\leq \max\bigg (1 - 1, 1 - \bigg \lceil \frac{h}{2^{\lceil \log h \rceil + 1}} \bigg \rceil\bigg ) = 0,
\end{align*}
indicating that $f_{\lceil \log h \rceil + 1  }$ is monotone.
\end{proof}

\subsection{Local implementation}
\label{subsection: local implementation non-refined}
In this section, we provide an LCA that gives membership query access to the output of~\Cref{alg:global}, and we analyze its complexity. To better explain how our LCA simulates~\Cref{alg:global}, we present it as a system of three LCAs, each parameterized by the iteration number $i$. \Cref{alg:path graphs} makes queries to the $i_{th}$ function $f_i$ and an all-neighbors oracle for $P$, and answers queries to the $i_{th}$ $k$-violation graph $B_i$. \Cref{alg:matchings} makes queries to $B_i$ and answers queries to a maximal matching $\lambda_i$ over it. \Cref{alg:functions} makes queries to $f_i$ and $\lambda_i$ and answers queries to $f_{i+1}$, which swaps the matched labels.


\begin{algorithm}[H]
\begin{algorithmic}
		\State Given: Target vertex $x$, \red{all-neighbors (immediate predecessor and successor)} oracle for $P$, membership query oracle $f_i$, height $h$, iteration number $i$. 
		\State Initialize $k \gets \lceil h / 2^{i+1} \rceil$
		\If {$f_i(x) = 1$}
		\State $S \gets$ the set of all successors of $x$ 
		\Else
		\State $S \gets$ the set of all predecessors of $x$ 
		\EndIf
		\State Compute longest path $\dist(x, y)$ for each $y \in S$ by dynamic programming
		\State Remove any $y$ from $S$ such that $\dist(x, y) < k$ or $f(y) = f(x)$
		\State \Return $S$
		
\end{algorithmic}
\caption{An LCA for undirected all-neighbors queries $B_i(x, P, f_i, h, i)$ }
\label{alg:path graphs}
\end{algorithm}

\begin{algorithm}[H]
\begin{algorithmic}
		\State Given: Target vertex $x$, undirected all-neighbors query oracle $B_i$, random seed $r$, and confidence parameter $\delta$.
		\State Call the algorithm described in \Cref{cor:matchings} with $B_i$ as the graph, $x$ as the target, and $r$ as the random seed \blue{ and $\Delta$ as the degree bound}. If $x$ is in the matching, return the vertex that it is matched to; otherwise return $x$.
\end{algorithmic}
\caption{An LCA for maximal matchings $\lambda_i(x, B_i, r, \delta)$}
\label{alg:matchings}
\end{algorithm}

\begin{algorithm}[H]
\begin{algorithmic}
		\State Given: Target vertex $x$, matching query oracle $\lambda_i$, membership query oracle $f_{i}$
		\State \Return $f_{i}(\lambda_i(x))$
\end{algorithmic}
\caption{An LCA for membership queries $f_{i+1}(x, \lambda_i, f_i)$}
\label{alg:functions}
\end{algorithm}

\subsubsection{Analysis of our implementation}
\label{subsubsection: analysis of our implementation}
Throughout this section, for any algorithm $A$, the notation $T(A)$ denotes the running time of $A$. 

\begin{claim}[Behavior of $B_i$]
\label{claim:gi-behavior}
For any $0 \le i \le \lceil \log h \rceil + 1$, $B_i$ provides all-neighbors query access to the $\lceil h/2^{i+1}\rceil$-violation graph of $P$ with respect to $f_i$. Furthermore, if each element of $P$ has at most $\Delta$ successors or predecessors, then $T(B_i) \le O(\Delta \cdot T(f_i))$.
\end{claim}
\begin{proof}
If $f_i(x) = 1$, then all neighbors of $x$ in the $\lceil h/2^{i+1}\rceil$-violation graph of $P$ are successors of $x$. Finding all the successors takes $O(\Delta)$ time by depth-first search. Computing $\dist(x, y)$ for each successor $y$ of $x$ takes $O(\Delta)$ time by standard dynamic programming techniques for finding longest paths in a DAG. Comparing $f_i(x)$ to $f_i(y)$ takes $O(\Delta)$ queries to $f_i$, and therefore $O(\Delta \cdot T(f_i))$ time. The case of $f_i(x) = 0$ is symmetric.
\end{proof}
%
\begin{claim}[Behavior of $\lambda_i$]
\label{claim:li-behavior}
For any $0 \le i \le \lceil \log h \rceil + 1$, and $\delta \in (0,1]$, $\lambda_i$ provides query access to a maximal matching over $B_i$ with probability $1 - \delta/(\lceil \log h \rceil + 1)$, using a random seed of length $\red{\poly(\Delta\log(N/\delta))}$. Furthermore, $T(\lambda_i) \le \red{\poly(\Delta \log(N/\delta)) \cdot T(B_i)}$.
\end{claim}
\begin{proof}
Let $\delta' = \delta/(\lceil \log h \rceil + 1)$.
By \Cref{cor:matchings}, $\lambda_i$ fails with probability at most $\delta'$, where the query complexity and the length of the random seed are each $\red{\poly(\Delta \log(N/\delta'))}$. Since $h$ is always at most $\Delta$, this is still $\red{\poly(\Delta\log(N/\delta))}$.
The claim follows from the fact that since the queries are made to $B_i$, each query takes time $T(B_i)$. 
\end{proof}
We now proceed with our proof of \Cref{thm: the local corrector for general posets}, which we restate here for convenience.
\main*

\begin{proof} 
Assume that for each $i \le \lceil \log h \rceil + 1$, the matching provided by $\lambda_i$ is maximal for $B_i$. Under this condition, $f_i$, $\lambda_i$, and $B_i$ implement LCAs for all the objects expected by \Cref{alg:global} (by \Cref{claim:gi-behavior} and \Cref{claim:li-behavior}). The correctness result of \Cref{lem:correctness} implies that our implementation of $f_{\lceil \log h \rceil + 1}$ is an LCA for a function with the properties claimed in this theorem. We will bound the running time and query complexity by a recurrence relation.
		
As a base case, we will let $T(f_0)$ be 1. We have three recurrences: $T(B_i) \le O(\red{\Delta} \cdot T(f_i))$, $T(\lambda_i) \le \red{\poly(\Delta \log(N/\delta))} \cdot T(B_i)$, and $T(f_{i+1}) \le O(T(f_i) + T(\lambda_i))$. To simplify: \\
\begin{align*} 
T(f_{i+1}) &\le O(T(f_i) + T(\lambda_i)) \\
		   &\le T(f_i) + \red{\poly(\Delta \log(N/\delta))} \cdot T(B_i) \\
		   &\le T(f_i) + \red{\poly(\Delta \log(N/\delta)) \cdot \Delta} T(f_i) \\
		   &\le \red{\poly(\Delta \log(N/\delta))} \cdot T(f_i) \\
\end{align*}
This recurrence resolves to 
\[T(f_i) = \red{(\Delta \log(N/\delta))^{O(i)}}\]
for both running time and query complexity. Letting $i = \lceil \log h \rceil + 1$, we have a total running time and query complexity of $\red{(\Delta \log(N/\delta))^{O(\log h)}}$. 

We initialize our algorithm with a random bit string of length $\red{\poly(\Delta \log (N/\delta))}$ as required by \Cref{claim:li-behavior}. Each call to $\lambda_i$ fails with probability $\delta/(\lceil \log h \rceil + 1)$; this gives a total failure probability of $\delta$. Therefore, with probability at least $1 - \delta$ all the matchings are maximal and our analysis holds.
\end{proof}

%% file: path_shortening_general.tex
Our algorithm works by finding a maximal matching over the $k$-violation graph $B_k$, and swapping the matched labels. We first claim that performing this swap reduces the \red{distance (length of the longest path)} between violated labels by at least $k$.


\begin{lemma}[Distance shortening lemma]\label{lemma: path shortening lemma}
Let $P$ be any poset. 
Let $f$ be a $\{0,1\}$-valued function over $P$ and $k$ be a positive integer. Let $B_{k}$ be as defined in \Cref{def:viol-graph}. Suppose one picks some maximal matching $M$ over $B_{k}$ and obtains a new function $f'$ as follows
$$
f'(x)=
\begin{cases}
f(x) &\text{if $x$ was not matched},\\
f(y) &\text{if $x$ was matched to some $y$}.
\end{cases}
$$
Then, we have
\begin{align*}
&\max_{(\bv, \bw) \in \viol_P(f')}\dist(v , w) \leq \\
&\max\left(k-1, \left(\max_{(\bv, \bw) \in \viol_P(f)}\dist(v , w)\right) - k \right).
\end{align*}
\end{lemma}
\begin{proof}
Let $\lambda: P \rightarrow P$ map $x$ to (i) $y$ if $x$ was mapped to some $y$ by $M$ (ii) $x$ itself otherwise. Note that $f'(x)=f(\lambda(x))$ and also $\lambda$ is one-to-one. 

Let $x,y$ in $P$ be such that $x\succ y$, $f(x)=0$ and $f(y)=1$. If $\dist(x,y) \geq k$, it cannot be the case that both $\lambda(x)=x$ and $\lambda(y)=y$, because then $M$ would not be a maximal matching since we could also match $x$ to $y$. Besides, note that if $\lambda(x) \neq x$ then $\lambda(x) \prec x$ and  $\dist(x,\lambda(x))\ge k$. Analogously, if $\lambda(y) \neq y$ then $\lambda(y) \succ y$ and  $\dist(y,\lambda(y)) \ge k$. Additionally, for any $a,b,c\in P$ if $a \succ b \succ c$ then $\dist(a,c)\ge \dist(a,b)+\dist(b,c)$, as there exists a path from $a$ to $c$ that is the union of the longest paths from $a$ to $b$ and $b$ to $c$.
Taking these observations together, we see that only following eight cases are possible:
\begin{enumerate}
    \item $\dist(x,y) \leq k-1$.
    \begin{enumerate}
        \item $\lambda(x) \succ \lambda(y)$ and \\ 
        $\dist(\lambda(x),\lambda(y)) \leq \dist(x,y) \leq k-1$.
        \item It is not the case that $\lambda(x) \succ \lambda(y)$.
    \end{enumerate}
    \item $\dist(x,y) \geq k$.
    \begin{enumerate}
    \item $\lambda(x)=x$ and $\dist(y,\lambda(y)\ge k$
    \begin{enumerate}
        \item $\lambda(x) \succ \lambda(y)$ and \\
        $\dist(\lambda(x),\lambda(y))\le \dist(x,y)-k$.
        \item It is not the case that $\lambda(x) \succ \lambda(y)$.
    \end{enumerate}
    \item $\lambda(y)=y$ and $\dist(x,\lambda(x))\ge k$
    \begin{enumerate}
        \item $\lambda(x) \succ \lambda(y)$ and \\
        $\dist(\lambda(x),\lambda(y))\le \dist(x,y)-k$.
        \item It is not the case that $\lambda(x) \succ \lambda(y)$.
    \end{enumerate}
    \item $\dist(x,\lambda(x))\ge k$ and $\dist(y,\lambda(y))\ge k$
    \begin{enumerate}
        \item $\lambda(x) \succ \lambda(y)$ and \\
        $\dist(\lambda(x),\lambda(y))\le \dist(x,y)-2k$.
        \item It is not the case that $\lambda(x) \succ \lambda(y)$.
    \end{enumerate}
    \end{enumerate}
\end{enumerate}
In the whole, if $\lambda(x) \succ \lambda(y)$ then $\dist(\lambda(x),\lambda(y))\leq \max(k-1,\dist(x,y)-k)$.

Now let's consider what happens after the swap. Let $v_0,w_0$ in $P$ maximize $\dist(v_0,w_0)$ subject to $v_0\succ w_0$, $f'(v_0)=0$ and $f'(w_0)=1$. Let $x=\lambda^{-1}(v_0)$ and $y=\lambda^{-1}(w_0)$.
Since $f'(v_0)=0$, we have $x \succeq v_0$ and since $f'(w_0)=1$, we have $y \preceq w_0$. Therefore, $x \succ y$. Also, $f(x)=f'(v_0)=0$, $f(y)=f'(w_0)=1$ and $\lambda(x)\succ\lambda(y)$. The conclusion of the previous paragraph tells us that $\dist(v_0,w_0)= \dist(\lambda(x),\lambda(y)) \leq \max(k-1,\dist(x,y)-k)$. Overall,
\begin{align*}
    &\max_{(v , w)\in\viol_P(f')}\dist(v , w) =  \dist(v_0,w_0) \\
    &\leq \max(k-1,\dist(x,y)-k) \\
    &\leq\max\left(k-1, \left(\max_{(v , w)\in\viol_P(f)}\dist(v , w)\right) - k \right),
\end{align*}
which finishes the proof.
\end{proof}

%% file: appendixb.tex
\section{Standard proofs}

\subsection{Improper agnostic learning algorithm for monotone functions from literature.}
\label{sec: best improper  agnostic learning algorithm from the literature.}
For the sake of completeness, in this appendix we elaborate on the references for \Cref{thm: bshouty-tamon agnostic}.

For a function $f: \{0,1\}^n \rightarrow \{0,1\}$ 
the \textbf{influence} $\text{Inf}(f)$ is defined as 
\[
\text{Inf}(f):=\sum_i 
\Pr_{x \sim \{0,1\}^n}[f(x) \neq f\left(x^{\oplus i})\right)],
\]
where $x^{\oplus i}$ is the bit-string $x$ with it's $i$th value flipped. Similarly, a more general notion of the \textbf{$L_1$-influence} of a real-valued function $f: \{0,1\}^n \rightarrow \R$  is defined as 
\[
\text{Inf}^1(f):=\sum_i 
\E_{x \sim \{0,1\}^n}[\abs{f(x) - f\left(x^{\oplus i})\right)}].
\]

\begin{lemma}[Corollary 3.3 of \cite{feldman_tight_2020}]
\label{lem:FKV}
For every $f: \{0,1\} \to \R$ such that $\E_{x \sim \{0,1\}^n}\abs{f(x)}^2 \le 1$ and every $\eps > 0$, there exists a multilinear polynomial $p$ of degree $d = \lceil \frac{2 \cdot \Inf^{1}(f)}{\eps} \log \frac{2}{\eps} \rceil$ such that $\E_{x \in \{0,1\}^n}\abs{f(x) - p(x)} \le \eps$.
\end{lemma}

The following is a standard lemma from \cite{bshouty1996fourier}. Also see the text \cite{odonnell_analysis_2021}.
\begin{lemma}[\cite{bshouty1996fourier}]
\label{lem:monotone functions have small influence}
For every monotone function $f: \{0,1\}^n \rightarrow \{0,1\}$ , it is the case that $\text{Inf}(f)\leq \sqrt{n}$.
\end{lemma}

Combining \Cref{lem:FKV} and \Cref{lem:monotone functions have small influence} we get immediately the following corollary:
\begin{corollary}
\label{corollary: monotone functions well-approximated in L1 norm by low-degree polynomials refined}
For every monotone $f: \{0,1\}^n \to \{0,1\}$ and $\eps > 0$, there exists a multilinear polynomial $p$ of degree $\lceil \frac{2 \cdot \sqrt{n}}{\eps} \log \frac{2}{\eps} \rceil$ such that 
\[\E_{x \in \{0,1\}^n}\abs{f(x) - p(x)} \le \eps.\]
\end{corollary}
Finally, the corollary above combinded with Remark 4 on page 6 of \cite{kalai_agnostically_2005}
yields \Cref{thm: bshouty-tamon agnostic} and standard failure probability reduction via repetition.

\subsection{Proof of \Cref{cor:matchings}}
\label{sec: proof of cor:matchings}

For a graph $G$, one forms the so-called line graph of $G$, denoted as $G'$, as follows: (i) the vertex set of $G'$ is the edge set of $G$, (ii) two vertices are connected in $G'$ if the corresponding edges in $G$ share a vertex. One sees immediately that maximal matchings on $G$ translate to maximal independent sets on $G'$, and vice versa.   
Therefore, one can use the LCA of \cite{GhaffariU19} (described here in Theorem \ref{thm: Ghaffari Uitto LCA theorem}) to get access to  a maximal independent set in $G'$, which will translate to a maximal matching on $G$. 

An all-neighbor query to $G'$ can be simulated via two all-neighbor queries to $G$. The graph $G'$ has at most $\Delta N$ vertices and degree at most $2\Delta$. Overall, this means that the query complexity and the run-time of the LCA are still \blue{$
\poly(\Delta, \log(N/\delta))$}, the length of the random bit-string is still \blue{$
\poly(\Delta, \log(N/\delta))$} and the failure probability is still at most $\delta$.

\subsection{Proof of \Cref{cor:poset tester}}
\label{sec: proof of cor:poset tester}

First of all, without loss of generality we can assume $\delta=2/3$, because error probability can be reduced via repetition.

Theorem \ref{thm: the local corrector for general posets} (via Remark \ref{remark: sorting implies local correction}) allows us to gain query access to $f_\text{mon}$, such that distance of $f$ to $f_\text{mon}$ is at most twice the distance of $f$ to monotonicity. Then, obtaining the values of both these functions on i.i.d. uniformly random points of $P$, we estimate $\norm{f-f_\text{mon}}$ up to error $0.005\epsilon$. Then, if the distance of $f$ to monotonicity is at least $\epsilon$, the distance $\norm{f-f_\text{mon}}$ will be also at least $\epsilon$, so the value of the estimate will be at least $0.995\epsilon$. On the other hand, if the distance of $f$ to monotonicity is at most $0.49\epsilon$, then $\norm{f-f_\text{mon}}$ will at most $0.98\epsilon$ and the value of the estimate will be at most $0.985\epsilon$. Overall, checking if the value of the estimate is greater than $0.99\epsilon$ we can see which of the two cases we are in.

For the estimation to succeed, we need to evaluate $f$ and $f_\text{mon}$ on $O\left(\frac{1}{\epsilon^2}\right)$ i.i.d. random elements of $P$\footnote{For the LCA of $f_\text{mon}$ we can set the overall success probability parameter to be $0.1$. A Chernoff bound and union bound argument then shows that overall success probability is at least $2/3$.}. Overall, this will take $\Delta^{O(\log h \log\log \Delta)}\left(\log\left(N\right)\right)^{O(\log h)}\frac{1}{\epsilon^2}$ queries and run-time.

\subsection{Proof of \Cref{thm: proper learning semi-agnostic}}
\label{sec: proof of thm: proper learning semi-agnostic}
The proof follows a pattern similar to the proof of Theorem \ref{thm: proper learning realizable}. 
The only modification to the algorithm in the proof of Theorem \ref{thm: proper learning realizable} is that in step 1 we obtain $g$ by using the agnostic improper learner of Theorem \ref{thm: bshouty-tamon agnostic} (instead of the learner of Theorem \ref{thm: bshouty-tamon realizable}). The accuracy parameter there will still be $\frac{\epsilon}{10}$. With probability at least $1/2^{n}$ both algorithms we use succeed, which we will assume henceforth. 

The run-time analysis remains the same, except the learner in Theorem \ref{thm: bshouty-tamon agnostic} now takes $n^{O\left(\frac{\sqrt{n}}{\epsilon}\log \frac{1}{\epsilon}\right)}$ samples and run-time (as opposed to
$n^{O\left(\frac{\sqrt{n}}{\epsilon}\right)}$ samples and run-time for the learner in Theorem \ref{thm: bshouty-tamon realizable}). 
Repeating the argument, the overall run-time and sample complexity is also $n^{O\left(\frac{\sqrt{n}}{\epsilon}\log \frac{1}{\epsilon}\right)}$.

Finally, we argue correctness. The function $f_\text{mon}$ computed by the circuit we output is again monotone by the same argument as in proof of Theorem \ref{thm: bshouty-tamon realizable}. Theorem \ref{thm: bshouty-tamon agnostic} tells us that function $f$ has generalization error of at most $\opt + \frac{\epsilon}{10}$. Also, let $f^*$ be a monotone function with the best available generalization error of $\opt$. This implies that, $f$ and $f^*$ disagree on at most $\left(2\opt+\frac{\epsilon}{10}\right)2^n$ elements\footnote{Specifically, for a random example-label pair $(x,y)$ ($x$ distributed uniformly) we have $\Pr \left[f(x) \neq f^*(x) \right] = \Pr \left[f(x)=y,~ f^*(x)\neq y \right] + \Pr \left[f(x) \neq y,~ f^*(x)= y \right] $, which can be upper-bounded by $\Pr\left[ f^*(x)\neq y \right]+\Pr\left[f(x)\neq y \right]$. Given the bounds we know for these probabilities, the bound on $\norm{f-f^*}$ follows.}.

Now, recall that the extra property of the LCA in Theorem \ref{thm: the local corrector for general posets} (noted in Remark \ref{remark: sorting implies local correction}) tells us that for any monotone function $q$ over $H^n_{\epsilon/10}$, its distance to $f_{\text{mon}}$ is at most its distance to $f$. Taking $q=f^*$, we get that $f_{\text{mon}}$ and $f^*$ can disagree only on at most $\left(2\opt+\frac{\epsilon}{10}\right)2^n$ elements of $H^n_{\epsilon/10}$. 
The number of $x\in \{0,1\}^n$ not in $H^n_{\epsilon/10}$ is at most $\frac{\epsilon}{10}2^n$, so overall $f_{\text{mon}}$ disagrees with $f^*$ on at most $\left(2\opt+\frac{\epsilon}{5}\right)2^n$ elements of $\{0,1\}^n$, in other words $\norm{f^*-f_{\text{mon}}} \leq 2\opt+\frac{\epsilon}{5}$. Via triangle inequality, the generalization error of $f_{\text{mon}}$ is at most the sum of (i) generalization error of $f^*$ and (ii) the distance between $f_{\text{mon}}$ and $f^*$. This means that the generalization error of $f_{\text{mon}}$ is at most $3\opt+\frac{\epsilon}{5}\leq 3\opt +\epsilon$, finishing the proof. 

%% file: appendixC.tex
\subsection{Refined local implementation.}
\label{sec: refined local implementation}
Here we explain how an improved run-time can be achieved for posets with additional characteristics. We shall assume that each element of the poset $P$ has at most $n$ immediate predecessors and at most $n$ immediate successors. Additionally, we assume that the poset $P$ is graded. 
This is achieved again by implementing the algorithm in \Cref{subsection: the global view}, but the run-time is somewhat faster than the one achieved in in \Cref{subsection: local implementation non-refined}. To be fully specific, \Cref{subsection: the global view} would give us a run-time of $n^{O(h \log h)}(\log (N/\delta))^{\log h}$ in this setting, which is here improved to $n^{O(h )}(\log (N/\delta))^{\log h}$.

Similarly to \Cref{subsection: local implementation non-refined}, we provide an LCA that gives membership query access to the output of~\Cref{alg:global}, and we analyze its complexity. This is again presented it as a system of three LCAs, each parameterized by the iteration number $i$. \Cref{alg:path graphs refined} makes queries to the $i_{th}$ function $f_i$ and an all-neighbors oracle for $P$, and answers queries to the $i_{th}$ $k$-violation graph $B_i$. \Cref{alg:matchings refined} makes queries to $B_i$ and answers queries to a maximal matching $\lambda_i$ over it. \Cref{alg:functions refined} makes queries to $f_i$ and $\lambda_i$ and answers queries to $f_{i+1}$, which swaps the matched labels.

\begin{algorithm}[H]
\begin{algorithmic}
		\State Given: Target vertex $x$, immediate predecessor and successor oracle for $P$, membership query oracle $f_i$, height $h$, iteration number $i$. 
		\State Initialize $k \gets \lceil h / 2^{i+1} \rceil$
		\If {$f_i(x) = 1$}
		\State $S \gets$ the set of all successors $y$ of $x$ such that $k \le \dist(x,y) \le 2k$.
		\Else
		\State $S \gets$ the set of all successors $y$ of $x$ such that $k \le \dist(x,y) \le 2k$.
		\EndIf
		\State Note: because the poset is graded, the collections of successors and predecessors above can be found via breath-first search.
		\State Remove any $y$ from $S$ such that $f(y) = f(x)$
		\State \Return $S$
		
\end{algorithmic}
\caption{An LCA for undirected all-neighbors queries $B_i(x, P, f_i, h, i)$ }
\label{alg:path graphs refined}
\end{algorithm}

\begin{algorithm}[H]
\begin{algorithmic}
		\State Given: Target vertex $x$, undirected all-neighbors query oracle $B_i$, random seed $r$, and confidence parameter $\delta$.
		\State Call the algorithm described in \Cref{cor:matchings} with $B_i$ as the graph, $x$ as the target, and $r$ as the random seed and $n^{\lceil h / 2^{i} \rceil}$ as the degree bound. If $x$ is in the matching, return the vertex that it is matched to; otherwise return $x$.
\end{algorithmic}
\caption{An LCA for maximal matchings $\lambda_i(x, B_i, r, \delta)$}
\label{alg:matchings refined}
\end{algorithm}

\begin{algorithm}[H]
\begin{algorithmic}
		\State Given: Target vertex $x$, matching query oracle $\lambda_i$, membership query oracle $f_{i}$
		\State \Return $f_{i}(\lambda_i(x))$
\end{algorithmic}
\caption{An LCA for membership queries $f_{i+1}(x, \lambda_i, f_i)$}
\label{alg:functions refined}
\end{algorithm}

\subsubsection{Analysis of our implementation}
Throughout this section, for any algorithm $A$, the notation $T(A)$ denotes the running time of $A$. 
This is the same convention used in \Cref{subsubsection: analysis of our implementation}. Furthermore, the following two claims are analogous to \Cref{claim:gi-behavior} and \Cref{claim:li-behavior} respectively.
\begin{claim}[Behavior of $B_i$]
\label{claim:gi-behavior refined}
For any $0 \le i \le \lceil \log h \rceil + 1$, suppose \[\max_{(v, w) \in \viol_P(f_i)} \dist(v, w) \leq  \lceil h/2^i  \rceil
.\] Then, $B_i$ provides all-neighbors query access to the $\lceil h/2^{i+1}\rceil$-violation graph of $P$ with respect to $f_i$. Furthermore, the degrees of all vertices in $B_i$ are bounded by $n^{\lceil h / 2^{i} \rceil}$ and
$T(B_i) \le O(n^{\lceil h / 2^{i} \rceil}\cdot T(f_i))$.
\end{claim}
\begin{proof}
If $f_i(x) = 1$, then all neighbors of $x$ in the $\lceil h/2^{i+1}\rceil$-violation graph of $P$ are successors of $x$. 
As, $\lceil h/2^{i}\rceil\leq 2 \lceil h/2^{i+1}\rceil=2k$, we see that initializing $S$ to have only elements of distance at most $2k$ does not leave out any neighbors of $x$ in the $\lceil h/2^{i+1}\rceil$-violation graph of $P$. Comparing with the definition of the $\lceil h/2^{i+1}\rceil$-violation graph of $P$, we see that the elements given by  \Cref{alg:path graphs refined} are precisely the neighbors of $x$ in the $\lceil h/2^{i+1}\rceil$-violation graph of $P$.

Finally, the bound of $n^{\lceil h / 2^{i} \rceil}$ on the degree of $B_i$ follows, because each element in $P$ has at most $n$ immediate predecessors or successors. The bound on the run-time follows from the bound on degree of $B_i$.
\end{proof}
%
\begin{claim}[Behavior of $\lambda_i$]
\label{claim:li-behavior refined}
For any $0 \le i \le \lceil \log h \rceil + 1$, and $\delta \in (0,1]$,
if degrees of all vertices in $B_i$ are bounded by $n^{\lceil h / 2^{i} \rceil}$, then
$\lambda_i$ provides query access to a maximal matching over $B_i$ with probability $1 - \delta/(\lceil \log h \rceil + 1)$, using a random seed of length $\poly(n^{\lceil h / 2^{i} \rceil}\log(N/\delta))$. Furthermore, the run-time $T(\lambda_i)$ is bounded by $\poly(n^{\lceil h / 2^{i} \rceil} \log(N/\delta)) \cdot T(B_i)$.
\end{claim}
\begin{proof}
Let $\delta' = \delta/(\lceil \log h \rceil + 1)$.
By \Cref{cor:matchings}, $\lambda_i$ fails with probability at most $\delta'$, where the query complexity and the length of the random seed are each $\poly(n^{\lceil h / 2^{i} \rceil } \log(N/\delta'))$.
The claim follows from the fact that since the queries are made to $B_i$, each query takes time $T(B_i)$. 
\end{proof}
We now proceed with our proof of \Cref{remark: run-time refinement}, that is to proving an overall run-time bound of $n^{O(h)}  (\log(N/\delta))^{\log h}$. 


\begin{proof} 
We first argue, using an induction over $i$, that $f_i$, $\lambda_i$, and $B_i$ implement LCAs for all the objects expected by \Cref{alg:global}
(with overall probability of at least $1-\delta$) 
. The base case $i=0$ is immediate. Suppose, this holds up to iteration $i$ (i.e. condition on this event). Then, by \Cref{claim:distance-invariant} we have \[\max_{(v, w) \in \viol_P(f_i)} \dist(v, w) \leq  \lceil h/2^i  \rceil,
\] so the premise of \Cref{claim:gi-behavior refined} holds. Now, one of the conclusions of \Cref{claim:gi-behavior refined} is that 
degrees of all vertices in $B_i$ are bounded by $n^{\lceil h / 2^{i} \rceil}$, which is the premise of \Cref{claim:li-behavior refined}.
Together, the conclusions of \Cref{claim:gi-behavior refined} and \Cref{claim:li-behavior refined}, imply that, with probability $1-\delta/(\lceil \log h \rceil + 1)$, $B_i$, $\lambda_i$ and $f_{i+1}$ implement the corresponding quantities expected by \Cref{alg:global}. 
Via a union bound over all $i$ we see that with overall probability of at least $1-\delta$ this indeed holds for all $i$.

We will bound the running time and query complexity by a recurrence relation.		
As a base case, we will let $T(f_0)$ be 1. We have three recurrences: $T(B_i) \le O(n^{\lceil h / 2^{i} \rceil} \cdot T(f_i))$, $T(\lambda_i) \le \poly(n^{\lceil h / 2^{i} \rceil} \log(N/\delta)) \cdot T(B_i)$, and $T(f_{i+1}) \le O(T(f_i) + T(\lambda_i))$. To simplify: \\
\begin{align*} 
T(f_{i+1}) &\le O(T(f_i) + T(\lambda_i)) \\
		   &\le T(f_i) + \poly(n^{\lceil h / 2^{i} \rceil}\log(N/\delta)) \cdot T(B_i) \\
		   &\le T(f_i) + \poly(n^{\lceil h / 2^{i} \rceil} \log(N/\delta)) \cdot n^{\lceil h / 2^{i} \rceil} T(f_i) \\
		   &\le \poly(n^{\lceil h / 2^{i} \rceil} \log(N/\delta)) \cdot T(f_i) \\
\end{align*}
Thus, we obtain
\begin{align*}
&T(f_i) = \prod_{i=0}^{\lceil \log h \rceil} \poly(n^{\lceil h / 2^{i} \rceil} \log(N/\delta)) \\
&=n^{O(h)}  (\log(N/\delta))^{\log h} 
\end{align*}

for both running time and query complexity. 

We note that our algorithm is initialized with a random bit string of length $n^{O(h)}\polylog(N/\delta)$, which is as required.
\end{proof}